\documentclass[letterpaper,12pt]{article}
\usepackage{amsmath, mathtools}
\usepackage{amssymb}
\usepackage{amsthm}
\usepackage{graphicx}
\usepackage{verbatim}
\usepackage{enumitem}
\usepackage{color}

\newtheorem{theorem}{Theorem}
\newtheorem*{theorem*}{Theorem}
\newtheorem{definition}{Definition}

\newtheorem{lemma}{Lemma}

\newtheorem{prop}{Proposition}

\def\PP{\mathbf{P}}     

\newcommand{\bsigma}{\boldsymbol{\sigma}}
\newcommand{\tmax}{t_{\textnormal{max}}}
\newcommand{\hsigma}{\hat{\sigma}}

\definecolor{Red}{rgb}{1,0,0}
\definecolor{Blue}{rgb}{0,0,1}
\def\red{\color{Red}}

\newcommand{\snote}[1]{{\red(S: #1)}}

\title{Pairwise sequence alignment at arbitrarily large evolutionary distance}

\author{
Brandon Legried\footnote{
School of Mathematics, Georgia Institute of Technology.} 
\and 
Sebastien Roch\footnote{
Department of Mathematics, University of Wisconsin--Madison. Corresponding author. Email: roch@math.wisc.edu.}
}

\date{\today}

\begin{document}

\maketitle

\begin{abstract}
Ancestral sequence reconstruction is a key task in computational biology. It consists in inferring a molecular sequence at an ancestral species of a known phylogeny, given descendant sequences at the tip of the tree. In addition to its many biological applications, it has played a key role in elucidating the statistical performance of phylogeny estimation methods. Here we establish a formal connection to another important bioinformatics problem, multiple sequence alignment, where one attempts to best align a collection of molecular sequences under some mismatch penalty score by inserting gaps. Our result is counter-intuitive: we show that perfect pairwise sequence alignment with high probability is possible in principle \emph{at arbitrary large evolutionary distances}---provided the phylogeny is known and dense enough. We use techniques from ancestral sequence reconstruction in the taxon-rich setting together with the probabilistic analysis of sequence evolution models involving insertions and deletions. 
\end{abstract}

\section{Introduction}

Ancestral sequence reconstruction (ASR) is a key task in computational evolutionary biology~\cite{liberles_ancestral_2007}. It consists in inferring a molecular sequence at an ancestral species of a known phylogeny, given descendant sequences at the tip of the tree. Numerous approaches are available for this task. Some are based on statistical models of sequence evolution on a tree, while others rely on combinatorial optimization formulations~\cite{semple_phylogenetics_2003, yang_molecular_2014}. In addition to its many biological applications, ASR has played a key role in elucidating the statistical performance of phylogeny estimation methods~\cite{mossel_impossibility_2003,mossel_phase_2004,roch_toward_2010,roch_phase_2017}. Here we establish a formal connection to sequence alignment. 

Rigorous analyses of the accuracy of ASR methods have been performed mainly in two asymptotic settings. In phylogenies of arbitrarily large depth, an achievable goal is to infer a sequence that is correlated site-by-site with the true ancestral sequence~\cite{steel_five_1995,ioffe_extremality_1996,mossel_recursive_1998,evans_broadcasting_2000}. In the taxon-rich setting, on the other hand, where the depth of the phylogeny is bounded as the number of taxa increases, consistent estimators are known to exist~\cite{gascuel_inferring_2010,roch_sufficient_2021}. That is, under conditions on the branching of the phylogeny around its root, the correct inference of a single site in the ancestral sequence can be guaranteed as the number of leaves goes to infinity. 

Most theoretical results in this area are derived under models of sequence evolution by single site substitutions. More complex models allowing for site insertions and deletions (indels) have also been considered~\cite{andoni_global_2012,ganesh_optimal_2019,fan_statistically_2020}. The star case, also known as trace reconstruction, has been the subject of much recent interest~\cite{holenstein_trace_2008,nazarov_trace_2017,holden_lower_2020,davies_approximate_2021,davies_reconstructing_2019}. See also \cite{THATTE200658,mitrophanov_convergence_2007,daskalakis2013alignment,Allman2015StatisticallyCK,fan2020impossibility} for rigorous analyses of indel models in other contexts, e.g., distance-based phylogeny reconstruction.

Indel models are closely related to another important bioinformatics problem, multiple sequence alignment (MSA), in which one attempts to best align a collection of molecular sequences under some mismatch penalty score by inserting gaps. In practice, MSA is a hard problem, especially at large evolutionary distances~\cite{rost_twilight_1999,chang_phylogenetic_2008}. While statistical approaches based on indel models have also been developed~\cite{Lunter2005}, commonly used approaches involve progressively aligning the given sequences up a guide tree, in what is reminiscent of ASR procedures~\cite{ranwez:hal-02535389}. In fact, many trace reconstruction and ASR methods under indels involve partial local alignments of sequences.

In this paper, we combine insights from ASR in the taxon-rich setting together with the probabilistic analysis of indel models to prove the first (as far as we know) rigorous guarantee for sequence alignment under an indel model on a phylogenetic tree. Our result is somewhat counter-intuitive: we show that perfect pairwise sequence alignment with high probability is in principle possible \emph{at arbitrary large evolutionary distances}---provided the phylogeny is known and dense enough. While such a condition may not be satisfied in real datasets, our analysis is a step towards a better theoretical understanding of MSA and its connections to ASR.

In a nutshell, we take advantage of the density of the phylogeny to estimate ancestral sequences with high probability along the path between two leaf sequences of interest, then reconstruct the history of mutations along the way. For the ASR step, we use a standard phylogenetic method known as parsimony, which seeks to use the smallest number of mutations possible to explain sequences at the leaves of a phylogeny. Rigorous analyses of parsimony are often challenging and have revealed the intricate, often unintuitive, behavior of the method~\cite{li_more_2008,fischer_maximum_2009,zhang_analyzing_2010,herbst_ancestral_2017,herbst_accuracy_2018}. In our taxon-rich setting, branching process results lead to rigorous guarantees on the ancestral reconstruction. 

The rest of the paper is organized as follows. In Section~\ref{section:main-results}, we state our main result after introducing some background. The alignment algorithm is presented in Section~\ref{section:alignment}. The proof is comprised of two parts: the ancestral estimation step is analyzed in Section~\ref{section:ancestral} and the alignment step is analyzed in Section~\ref{section:one-mutation}.

\section{Background and main result}
\label{section:main-results}

In this section, we state our main result. First, we introduce the model of sequence evolution we use here as well as the multiple sequence alignment problem.

\subsection{Definitions} 

We consider the TKF91 insertion-deletion (indel) sequence evolution model. Technically, we use a slight variant of the TKF91 model defined in~\cite{Thorne1991}, where we only allow an alphabet with two letters $0$ and $1$ to simplify the analysis and its presentation.  Our results extend naturally to more general settings.

\begin{definition}[TKF91  model: two-state version]
\label{Def:BinaryIndel}
Consider the following Markov process $\mathcal{I} = \{\mathcal{I}_{t}\}_{t \geq 0}$ on the space $\mathcal{S}$ of binary digit sequences together with an \textbf{immortal link $``\bullet"$}, that is,
\begin{equation*}\label{S}
		\mathcal{S} := ``\bullet" \otimes \bigcup_{M\geq 1} \{0,1\}^M,
\end{equation*}
where the notation above indicates that all sequences begin with the immortal link.  Positions of a sequence, except for that of the immortal link,  are called \textbf{sites} or \textbf{mortal links}. Let $(\eta,\lambda,\mu) \in (0,\infty)^{3}$ and $(\pi_0,\pi_1) \in [0,1]^2$ with $\pi_0 + \pi_1 = 1$ be given parameters.  The continuous-time dynamics are as follows:  If the current state is the sequence $\vec{x} \in \mathcal{S}$, then the following events occur independently:

\begin{itemize}
	\item \emph{Substitution:}  Each site 
	is substituted independently at rate $\eta > 0$.  When a substitution occurs, the corresponding digit is replaced by $0$ and $1$ with probabilities $\pi_0$ and $\pi_1$, respectively.
	\item \emph{Deletion:}  Each site
	is removed independently at rate $\mu$.
	\item \emph{Insertion:}  Each site, as well as the immortal link,  gives birth to a new digit independently at rate $\lambda$.  When a birth occurs, the new site is added immediately to the right of its parent site.  The newborn site has digit $0$ and $1$ with probabilities $\pi_0$ and $\pi_1$, respectively.
\end{itemize}
\end{definition}

We run this process on a rooted metric tree as follows. Consider a \textit{rooted binary tree} $T = (V,E,\rho,\mathbf{t})$ with vertices $V$, edges $E$, root $\rho$, and edge lengths $\mathbf{t} = \{t_e\}_{e \in E}$ (in time units). We restrict ourselves to ultrametric trees, that is, the sum of edge lengths from root to leaf is the same for every leaf. We refer to this common quantity as the \textbf{depth} of the tree and denote it by $h$. The rooted metric tree $T$ is then indexed by all points along the edges of $T$.  The root vertex has an initial sequence $\sigma_{\rho} \in \mathcal{S}$. With an initial sequence $\sigma_u \in \mathcal{S}$, the TKF91 process is recursively performed on each descending edge $e = (u,v)$ over the time interval $[0,t_e]$ to obtain another sequence $\sigma_v \in \mathcal{S}$. Processes running along descending edges of $u$ are independent, conditioned on state $\sigma_u$ at $u$.  We refer to the full process as the \textbf{(two-state) TKF91 process on tree $T$}.

For any sequence $\sigma \in \mathcal{S}$, let $|\sigma|$ be the length of the sequence, and let  $|\sigma|_{0}$ and $|\sigma|_{1}$ be the number of $0$'s and $1$'s in the sequence, respectively.  The stationary distribution of the sequence length $|\sigma| = M$ is known~\cite{Thorne1991} to be 
\begin{equation} 
\label{eq:LengthStationary}
    \gamma_{M} = \left(1 - \frac{\lambda}{\mu}\right)\left(\frac{\lambda}{\mu}\right)^{M}, \qquad M \in \mathbb{Z}_+,
\end{equation} 
provided $\mu > \lambda$.
We assume that the root sequence $\sigma_\rho$ follows its stationary distribution.  That is, $|\sigma_{\rho}|$
is distributed according to $\gamma_{M}$ and its sites are i.i.d.~in $\{0,1\}$ with respective probabilities $\pi_0$ and $\pi_1$.
Stationarity of $\sigma_{\rho}$ implies stationarity of the TKF91 process throughout the tree. We assume from now on that $\mu > \lambda$ and that stationarity holds. 

\paragraph{Some notation}
Later on, we will need the following notation.
For a sequence $\sigma \in \mathcal{S}$, 
let $\mathcal{S}_{s}(\sigma)$,  $\mathcal{S}_{d}(\sigma)$, and $\mathcal{S}_{i}(\sigma)$ be 
the sequences that differ from $\sigma$ 
respectively by a single substitution, a single 
deletion, and a single insertion.
Observe that these sets are disjoint as
the sequence lengths in each necessarily differ.
Further, let $\mathcal{S}_1(\sigma) = \mathcal{S}_{s}(\sigma) \cup \mathcal{S}_{d}(\sigma) \cup \mathcal{S}_{i}(\sigma)$ be the sequences obtained by performing a single mutation on $\sigma$,
and define
\begin{equation}\label{eq:lambdastar}
\lambda^{\ast}(\sigma) = \sum_{\tau \in \mathcal{S}_1(\sigma)}Q(\sigma,\tau) = \lambda (|\sigma| + 1) + \mu |\sigma| + \eta \pi_1 |\sigma|_0 + \eta \pi_0 |\sigma|_1
\end{equation}
as the total rate under the TKF91 process of moving
away from $\sigma$, where $Q(\sigma,\tau)$ is the rate
at which the TKF91 process on an edge jumps from $\sigma$ to $\tau$. Formula~\eqref{eq:lambdastar} is derived formally in the appendix.

\subsection{Multiple sequence alignment}

To compare sequences descending from a common ancestor through substitutions, insertions and deletions, it is natural to attempt to align them as best as possible, that is, to construct a multiple sequence alignment.
\begin{definition}[Multiple sequence alignment]
For any integer $n \geq 1$ and sequences $\bsigma = (\sigma_{v_1},\ldots,\sigma_{v_m}) \in \mathcal{S}^m$ at points $v_1,\ldots,v_m \in T$, a \textbf{multiple sequence alignment} (or pairwise alignment when $m=2$) is a collection of sequences $\mathbf{a}(\bsigma) = (a_1(\bsigma),\ldots,a_m(\bsigma))$ whose entries come from $\{0,1,-\}$ ($-$ is called a \textit{gap}) such that: \begin{itemize}
    \item the lengths satisfy $$|a_1(\bsigma)| = |a_2(\bsigma)| = \cdots = |a_m(\bsigma)| \geq \max\{|\sigma_{v_1}|,|\sigma_{v_2}|,\ldots,|\sigma_{v_m}|\},$$
    \item no corresponding entries of $a_1(\bsigma),\ldots,a_m(\bsigma)$ all equal $-$, and
    \item removing $-$ from $a_i(\bsigma)$ yields $\sigma_{v_i}$ for all $i \in \{1,2,\ldots,m\}$.
\end{itemize} A multiple sequence alignment can be expressed as an $m \times |a_1(\bsigma)|$ matrix where the rows are the sequence alignments and where no column consists of all gaps. If $m=2$, the alignment is referred to as pairwise.
\end{definition}
\noindent More generally, a multiple sequence alignment procedure may take as input further auxiliary information (beyond the sequences to be aligned), such as a tree or sequences at other points of the tree. Our alignment algorithm (see Section~\ref{section:alignment}) will indeed use additional information.

Two sites, one from one sequence and the other from another sequence, are said to be \textbf{homologous} provided they descend from a common site in their most recent common ancestral sequence \textit{only through substitutions} under the evolutionary process on the tree. 
A \textbf{true} multiple sequence alignment 
is one that places homologous sites in the same column and non-homologous sites in different columns. We note however that certain homology relationships are unknowable a priori: for example, if in the course of evolution a $0$ is inserted in a sequence next to another $0$, which of them descends from the ancestral $0$ is arbitrary. Here we take the convention that a repeated site is always inserted at the beginning of a run; and that similarly a repeated site is always deleted at the beginning of a run. 

\subsection{Statement of main result}

The following theorem states that it is possible to construct with high probability a true pairwise alignment of the sequences at two arbitrary leaves $v$ and $w$ of a phylogeny as long as the maximal branch length is sufficiently small.  
\begin{theorem}[Main Result]
\label{thm:main}
Fix $\eta,\mu,\lambda \in (0,\infty)$, the substitution, deletion, and insertion rates under the TKF91 model. There is a polynomial-time alignment procedure $A$ such that 
for any tree depth $h > 0$ and any failure probability $\varepsilon > 0$, there exists a maximum branch length $\tmax := \tmax(h,\varepsilon) > 0$ such that
the following property holds. For any \textit{rooted binary tree} $T = (V,E,\rho,\mathbf{t})$ with vertices $V$, edges $E$, root $\rho$, and edge lengths $\mathbf{t} = \{t_e\}_{e \in E}$, assume that the leaves $\partial T = \{\ell_i\}_{i=1}^{n}$
are ordered from left to right in a planar realization of $T$, and let
$v = \ell_1$ and $w=\ell_n$. Then the alignment procedure applied to the sequences $\sigma_{\ell_1},\sigma_{\ell_2},\ldots,\sigma_{\ell_n}$ outputs a true pairwise alignment of $\sigma_v$ and $\sigma_w$ with probability at least $1- \varepsilon$, provided that $t_e \leq \tmax$ for all edges $e \in E$.
\end{theorem}
\noindent Note that the tree depth $h$ is arbitrary. The alignment procedure, which is described in Section~\ref{section:alignment}, takes as input leaf sequences at the leaves of $T$ as well as $T$ itself. 

\paragraph{Extensions}
While we assume above that the rate of substitution is the same throughout the tree, our proof still goes through if the parameter $\eta$ is merely an upper bound on that rate across edges. Similarly, our two-state assumption and the details of the substitution model
do not play a critical role in the proof. We make these assumptions to simplify the presentation. 

\section{Alignment algorithm}
\label{section:alignment}

In this section, we describe the alignment procedure of 
Theorem~\ref{thm:main}. We emphasize that this algorithm is not meant to be practical, but rather serve as a proof of our main result. 

\subsection{Overview of full alignment algorithm}

We introduce the following alignment algorithm $A$ which takes as input a rooted metric tree $T$, two distinguished leaves $v$ and $w$, all leaf sequences, and a pre-processing parameter $\delta_1$. We take $\delta_1$ to satisfy $\tmax \leq \delta_1 \leq h$.
The algorithm outputs a 
pairwise alignment for the sequences at $v$ and $w$. 

There is a unique path between $v$ and $w$ that we henceforth call the \textit{backbone}. We let $B$ be the number of \textit{non-root} vertices on the backbone. Then $v = x_1$ and $w = x_B$ and the other \textit{non-root} backbone vertices are in order $x_2,...,x_{B-1}$. For some parts of the algorithm and analysis, it will be convenient to use an alternative numbering of the backbone vertices---away from the root, numbering the left side, then numbering the right side. Specifically, let $x_1,\ldots,x_{B^-}$ be the backbone vertices on the same side of the root as $x_1$. Let $x_{B^-+1},\ldots,x_B$ be the backbone vertices on the same side of the root as $x_B$ and let $B^+$ be their number. Then we set 
$$
\tilde{x}^-_i := x_{B^- - (i-1)}, \qquad i =1,\ldots, B^-
$$ 
and
$$
\tilde{x}^+_i := x_{B^- + i}, \qquad i = 1,\ldots, B^+. 
$$
Notice in particular that $\tilde{x}^-_1$ and $\tilde{x}^+_1$ are the children of the root and $\tilde{x}^-_{B^-} = x_1$ and $\tilde{x}^+_{B^+} = x_B$.

\begin{figure}[t]
    \centering
    \includegraphics[scale=1.2]{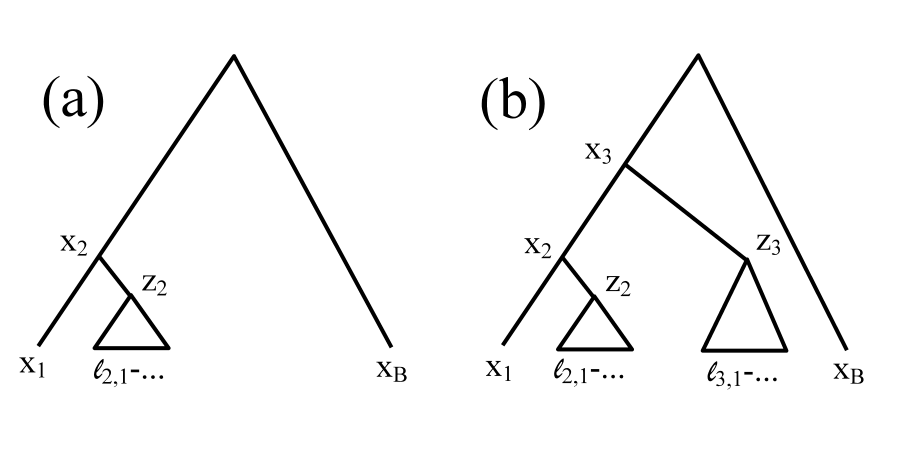}
    \caption{(a) Tree $T$ with leaf sequences $\sigma_1 = \sigma_{x_1},\sigma_2= \sigma_{\ell_{2,1}},..., \sigma_{\ell_{i_2}},\sigma_B = \sigma_{x_B}$. (b) Tree $T$ with leaf sequences $\sigma_1 = \sigma_{x_1},\sigma_2 = \sigma_{x_{2,1}},..., \sigma_{\ell_{i_2}},\sigma_{3} = \sigma_{x_{3,1}},...,\sigma_{\ell_{i_3}},\sigma_B$.}
    \label{fig:AlignmentProc1-12}
\end{figure}
We now describe the main steps of the algorithm. Some details will be given in the following subsections.
Figure~\ref{fig:AlignmentProc1-12} illustrates part of this algorithm at a high level.

We start with a pre-processing step.
\begin{itemize}
    \item \textbf{Pre-processing: backbone sparsification.} We first construct a subtree $T'$ by pruning some backbone vertices and their descendants. Initialize $T' := T$. Then, for $o = -,+$ and for $k = 1,...,B^o-1$,
    \begin{enumerate}
        \item Check whether the vertex $\tilde{x}^o_k$ is a vertex in the tree $T'$.  If not, do nothing.
        \item\label{item:preproc-ell} If $\tilde{x}^o_k$ is in the tree $T'$, find the minimal $\ell \geq 1$ such that the distance (accounting for edge lengths) between $\tilde{x}^o_k$ and $\tilde{x}^o_{k+\ell}$ is at least $\delta_1$. Observe that, by assumption, the distance between $\tilde{x}^o_k$ and $\tilde{x}^o_{k+\ell}$ is necessarily at most $\delta_1 + \tmax \leq 2 \delta_1$.
        \item Remove the vertices $\tilde{x}^o_{k+1},...,\tilde{x}^o_{k+\ell-1}$, except $\tilde{x}^o_{B^o}$, and all of their off-backbone descendants from the tree $T'$ (if they exist).
    \end{enumerate} 
\end{itemize}
The result is a tree where the distance between consecutive vertices on the backbone is in $[\delta_1, 2 \delta_1]$ (by the observation in Item~\ref{item:preproc-ell}), with the possible exception of the children of the root and the last pair on each side of the root all of whose distances are in $(0,2\delta_1]$. To simplify the notation, we re-assign $T$ to be this new rooted metric tree and we re-assign $x_1, x_2, \ldots,x_B$ to be the backbone vertices on this tree (with an updated value for $B$ and updated alternative numbering $\tilde{x}^o_k$ for $o = -,+$ and $k=1,\ldots,B^o$).

We then proceed with the alignment algorithm, which consists of two main steps both proceeding along the backbone:
\begin{enumerate}
    \item \textbf{Ancestral estimation:} We infer the ancestral sequences at the backbone vertices as follows. For $k = 2,...,B-1$:
    \begin{enumerate}
        \item For the child vertex $z_k$ of $x_k$ that is off the backbone, infer the sequence $\hsigma_{z_k}$ at $z_k$ using the Fitch method~\cite{Fitch71} (described below in Section~\ref{section:fitch}) applied to the subtree rooted at $z_k$.
        \item Set $\hsigma_{x_k}$ equal to $\hsigma_{z_k}$.
    \end{enumerate}
    \item \textbf{Recursive alignment:} Now that the sequences at the non-root backbone vertices $\{x_k\}_{k=1}^{B}$ have been estimated, we construct a multiple sequence alignment sequentially, starting from $x_1$, going to $x_2$, and ending at $x_{B-1}$ and $x_B$. This stepwise alignment procedure is described in Section~\ref{section:stepwise} below. If the inferred sequences of successive backbone vertices are not at most one mutation apart, then we terminate the algorithm with no output. 
    Else, a pairwise sequence alignment is produced for vertices $v = x_1$ and $w = x_B$. 
\end{enumerate} 

We will show in Proposition~\ref{prop:Fitch2} below that, with high probability, 
$\hsigma_{z_k} = \sigma_{z_k}$ for all $k = 2,\ldots, B-1$. 
We will then show in Proposition~\ref{prop:Intersection} below that the above stepwise alignment outputs a true pairwise alignment with high probability. 


\subsection{Ancestral sequence reconstruction}
\label{section:fitch}

We briefly describe below the ancestral sequence reconstruction subroutine. Note that we use the Fitch method for the convenience of its analysis, but other methods could also be used.
\begin{definition}[Fitch estimator]
Let $T = (V,E)$ be a finite binary rooted tree with root $z$ and leaf set $\partial T \subset V$ with given leaf sequences $(\sigma_{\ell})_{\ell \in \partial T}$.  For any leaf vertex $\ell$, define $\hat{S}_\ell \subset \mathcal{S}$ to be the subset $\hat{S}_\ell = \{\sigma_\ell\}$.  For each non-leaf vertex $v$ with children $v_1$ and $v_2$, define $\hat{S}_v \subset \mathcal{S}$ recursively to be \begin{align*}
    \hat{S}_v = \begin{cases}
        \hat{S}_{v_1} \cap \hat{S}_{v_2} & \textnormal{if} \ \hat{S}_{v_1} \cap \hat{S}_{v_2} \ne \emptyset \\
        \hat{S}_{v_1} \cup \hat{S}_{v_2} & \textnormal{otherwise}.
    \end{cases}
\end{align*} Then define the \textit{Fitch estimator} $\hat{\sigma}_z$ of $\sigma_z$ to be a uniformly chosen member of $\hat{S}_z$.
\end{definition} 
\noindent An analysis of this method in our setting is provided in Proposition~\ref{prop:Fitch2} below.

\subsection{Stepwise alignment}
\label{section:stepwise}

In this section, we describe the stepwise alignment subroutine. It is based on the assumption that along the backbone (of the pruned tree): 
\begin{enumerate}
\item[(i)] the sequences have been correctly inferred; and 
\item[(ii)] consecutive ones differ by at most one mutation. 
\end{enumerate}
We establish these facts in Propositions~\ref{prop:Fitch2} and~\ref{prop:Intersection} below.
In these circumstances, we show that homologous sites can be traced 
(up to the convention we described earlier).
We will construct a sequence of alignments
$\mathbf{a}^2$, $\mathbf{a}^3$, etc. 
We first describe the alignment of two sequences, then the alignment of alignments, and so on.



Given two sequences $\hat\sigma, \hat\tau$ satisfying the assumptions (i) and (ii) above, there are three possible cases:
\begin{enumerate}[label=(\Alph*)]
    \item If $\hat\sigma=\hat\tau$, then a true 
    alignment 
    is obtained by setting $a^{2}_1(\hat\sigma,\hat\tau) = \hat\sigma$ and $a^{2}_2(\hat\sigma, \hat\tau) = \hat\tau$, corresponding to no mutation. 
    \item If $|\hat\sigma| = |\hat\tau|$ but $\hat\sigma$ and $\hat\tau$ agree on all sites except one, 
    then a true alignment 
    is obtained by setting $a^{2}_1(\hat\sigma,\hat\tau) = \hat\sigma$ and $a^{2}_2(\hat\sigma,\hat\tau) = \hat\tau$, corresponding to exactly one substitution between the sequences.
    \item If $|\hat\sigma| = |\hat\tau| + 1$ (or vice versa) and there exists $j \in \{1,2,...,|\hat\tau|\}$ and $\hat\sigma_{\textnormal{ins}} \in \{0,1\}$ such that \begin{align*}
    \hat\sigma_i = \begin{cases}
    \hat\tau_i & i < j \\
    \hat\sigma_{\textnormal{ins}} & i = j \\
    \hat\tau_{i-1} & i > j.
    \end{cases}
\end{align*} 
As we discussed before, the location of the insertion cannot be determined from the sequences alone.
For example, if $\hat\sigma$ and $\hat\tau$ are separated by an insertion so that they are given by 
\begin{align*}
    \hat\tau &= (0,1,0,1,0,1,0,0,0,0,0,1,0)\\
    \hat\sigma &= (0,1,0,1,0,1,0,0,0,0,0,0,1,0),
\end{align*} 
we cannot tell which site gave birth to the new $0$ to obtain $\hat\sigma$. So we assume by convention that $j$ is the minimal choice possible.
Then a true 
alignment is obtained by setting 
$a^{2}_1(\hat\sigma,\hat\tau) = \hat\sigma$ and for $i=1,\ldots,|\hat\tau|+1$ 
\begin{align*}
    a^{2}_2(\hat\sigma,\hat\tau)_i =
        \begin{cases}
            \hat\tau_i & i < j \\
            - & i = j \\
            \hat\tau_{i-1} & i > j,
        \end{cases}
\end{align*}
corresponding to
a single site $\hat\sigma_{\textnormal{ins}}$ being inserted into the sequence $\hat\tau$ to the left of the $j$th site to obtain $\hat\sigma$. 
\end{enumerate}


In fact, we will need to align 
alignments along the backbone, rather than sequences. Suppose we have sequences $\hat\sigma_{x_1},\hat\sigma_{x_2},...,\hat\sigma_{x_B}$ and successive pairs $\{\hat\sigma_{x_1},\hat\sigma_{x_2}\},\{\hat\sigma_{x_2},\hat\sigma_{x_3}\},...,\{\hat\sigma_{x_{B-1}},\hat\sigma_{x_B}\}$ each satisfy exactly one of the cases (A), (B), or (C). (We terminate without output if the assumptions do not hold.) Then we recursively construct a multiple sequence alignment as follows.
To simplify the notation, we let
$
\hat\sigma_{1:k}
= (\hat\sigma_{x_1},\ldots,\hat\sigma_{x_k}).
$
\begin{enumerate}
    \item Given $\hat\sigma_{x_1}$ and $\hat\sigma_{x_2}$, let $a^{2}_1(\hat\sigma_{1:2})$ and $a^{2}_2(\hat\sigma_{1:2})$ be the pairwise alignment constructed above.
    \item For $k = 3,...,B$: \begin{enumerate}
        \item We are given a multiple alignment $a^{k-1}_1(\hat\sigma_{1:k-1}),\ldots,a^{k-1}_{k-1}(\hat\sigma_{x_{1:k-1}})$ of the sequences $\hat\sigma_{x_1},\ldots,\hat\sigma_{x_{k-1}}$, and a new sequence $\hat\sigma_{x_k}$ that is at most one mutation away from $\hat\sigma_{x_{k-1}}$.
        \item The sequences $\hat\sigma_{x_{k-1}}$ and $\hat\sigma_{x_k}$ satisfy one of the three cases (A), (B) or (C) by assumption, so their alignment $a^{k}_{k-1}(\hat\sigma_{x_{1:k}})$ and $a^{k}_{k}(\hat\sigma_{1:k})$ (within the larger multiple sequence alignment) will differ by at most one entry similarly to the sequence case above. The full alignment is defined as follows:
        \begin{itemize}
            \item If $\hat\sigma_{x_{k-1}} = \hat\sigma_{x_k}$, then set $a^{k}_k(\hat\sigma_{1:k})$ to be equal to $a^{k-1}_{k-1}(\hat\sigma_{1:k-1})$ and $a^{k}_i(\hat\sigma_{1:k})$ to be equal to $a^{k-1}_i(\hat\sigma_{1:k-1})$ for all $i < k$.
            \item If $\hat\sigma_{x_{k-1}}$ and $\hat\sigma_{x_k}$ have equal length and disagree at a single segregating site, 
            set $a^{k}_i(\hat\sigma_{1:k})$ to $a^{k-1}_i(\hat\sigma_{1:k})$ for all $i \leq k-1$.  Each entry of $a^{k}_k(\hat\sigma_{1:k})$ is set to the corresponding entry of $a^{k}_{k-1}(\hat\sigma_{1:k})$, except for the segregating site. If the latter occurs at position $i$ within $a^{k}_{k-1}(\hat\sigma_{1:k})$, then we set $a^{k}_k(\hat\sigma_{1:k})_i$ to $a^{k}_{k-1}(\hat\sigma_{1:k})_i + 1 \ (\textnormal{mod}\ 2)$.
            \item If 
            $\hat\sigma_{x_k}$ has one more site than $\hat\sigma_{x_{k-1}}$, then an insertion has occurred and the inserted site in $\hat\sigma_{x_k}$ cannot be ancestral to any site in $\hat\sigma_{x_1},\ldots,\hat\sigma_{x_{k-1}}$.  So the inserted site in $\hat\sigma_{x_{k}}$ must correspond to a gap in all previous sequences.  More specifically, if the site $\hat\sigma_{\textnormal{ins}}$ is inserted 
            to the left of position 
            $j^{\ast} \in \{1,\ldots,|a^{k-1}_{k-1}(\hat\sigma_{1:k-1})|\}$
            in the $(k-1)$-st sequence \textit{in the previously constructed alignment $a^{k-1}_{k-1}(\hat\sigma_{1:k-1})$} (where $j^{\ast}$ is the minimal such choice)
            then set 
            \begin{align*}
                a^k_k(\hat\sigma_{1:k})_{i} = \begin{cases}
                    a^{k-1}_{k-1}(\hat\sigma_{1:k-1})_{i} & 1 \leq i < j^{\ast} \\
                    \hat\sigma_{\textnormal{ins}} & i = j^{\ast} \\
                    a^{k-1}_{k-1}(\hat\sigma_{1:k-1})_{i-1} & j^{\ast} < i \leq |a^{k-1}_{k-1}(\hat\sigma_{1:k-1})| + 1
                \end{cases}
            \end{align*} 
            and for all $\ell \leq k-1$ 
            \begin{align*}
                a^{k}_\ell(\hat\sigma_{1:k})_{i} = \begin{cases}
                    a^{k-1}_\ell(\hat\sigma_{1:k-1})_{i} & 1 \leq i < j^{\ast}  \\
                    - & i = j^{\ast} \\
                    a^{k-1}_\ell(\hat\sigma_{1:k-1})_{i-1} & j^{\ast} < i \leq |a^{k-1}_{k-1}(\hat\sigma_{1:k-1})| + 1.
                \end{cases}
            \end{align*}
            \item The case where 
            $\hat\sigma_{x_k}$ has one fewer site than $\hat\sigma_{x_{k-1}}$ is handled symmetrically.
        \end{itemize}
    \end{enumerate}
    \item Output the pairwise alignment $(a^{B}_1(\hat\sigma_{1:B}), a^{B}_B(\hat\sigma_{1:B}))$ \emph{after removing all columns with only gaps}.
\end{enumerate}

\subsection{Theoretical guarantee}

We establish the two claims below in the next sections.
\begin{prop}[Correctness of alignment]
\label{prop:correct-align}
Let $T$ be the output of the pre-processing step and let $x_1,\ldots,x_B$ be the resulting backbone vertices. Then the alignment algorithm produces
a true pairwise alignment of $\sigma_{x_1}$ and $\sigma_{x_B}$ provided that:
\begin{enumerate}
    \item (Correctness of ancestral estimation) For $k=2,\ldots,B-1$, 
    $\hat\sigma_{x_k} = \sigma_{x_k}$.
    \item (One-mutation condition) Successive pairs of true backbone sequences
    $$
    \{\sigma_{x_1},\sigma_{x_2}\},\{\sigma_{x_2},\sigma_{x_3}\},\ldots,\{\sigma_{x_{B-1}},\sigma_{x_B}\},
    $$ 
    are at most one mutation away.
\end{enumerate}
\end{prop}

\section{Correctness of ancestral estimation}
\label{section:ancestral}

In this section, we analyze the ancestral sequence estimation step. The analysis proceeds by coupling the TKF91 process with a percolation process. Roughly, we say that an edge is \textbf{open} if the sequence does not change along it under the TKF91 process.  We will show that, provided the edge lengths are short enough, the open cluster of the root forms a fairly ``dense'' subtree with high probability. The latter property will lead to a correct reconstruction by the Fitch method.  





\subsection{The percolation process}

Consider again the backbone vertices $\{x_k\}_{k=1}^{B}$ and the off-backbone child vertices $\{z_k\}_{k=2}^{B-1}$. For each $k=2,\ldots,B-1$ separately, we couple the sequence evolution process on the subtree descending from $z_k$ with a simpler percolation process. 

We will need some notation. We denote by $T_k$ the subtree of $T$ (after pre-processing) rooted at $z_k$. For two sequences $\sigma, \tau$, 
we let $P_t(\sigma,\tau)$ be the probability under the TKF91 model on an edge
that, started at $\sigma$, the state is $\tau$ after time $t$. Similarly, we let $\tilde{P}_t(\sigma,\tau)$ be the same probability \textit{conditioned on not being at state $\sigma$ at time $t$}.  

It will be convenient to work on an infinite tree. Specifically, let $\overline{T}_k$ be the completion of $T_k$ into an infinite binary tree where new edges have length $0$. We now describe the coupling:
\begin{itemize}
    \item \textbf{The percolation process on $\overline{T}_k$.} We condition on the sequence $\sigma_{z_k}$ at $z_k$, the root of $T_k$. 
For an edge $e$ on $\overline{T}_k$, let $t_e$ be its length and set
$$
\zeta^{(k)}_{e} = P_{t_e}(\sigma_{z_k}, \sigma_{z_k}),
$$
that is, the probability that the sequence does not change along edge $e$ if started at $\sigma_{z_k}$. We then perform percolation
on $\overline{T}_k$ with probabilities $\zeta^{(k)}_e$:
for each edge $e$, it is open independently with probability $\zeta^{(k)}_e$. Let $\mathcal{C}_k$ be the resulting open cluster including $z_k$ (i.e., all vertices of $\overline{T}_k$ that can be reached from $z_k$ using only open edges). 

    \item \textbf{The joint process on $\overline{T}_k$.} For each vertex $v$ in $\mathcal{C}_k$, set $\sigma_v = \sigma_{z_k}$. For each descendant $w$ of a vertex $v \in \mathcal{C}_k$ that is not itself in $\mathcal{C}_k$, assign a sequence to $w$ taken from the conditional distribution $\tilde{P}_{t_e}(\sigma_{z_k},\,\cdot\,)$ where $e = (v,w)$.  
    For each remaining vertex, we run the TKF91 process recursively from the states already assigned. Note that the edges of length $0$ added in the completion of $T_k$ simply entail copying the sequences at the leaves of $T_k$ to all their descendants. 
\end{itemize} 

We will be interested in the properties of the cluster $\mathcal{C}_k$. 
Let $\overline{T}_k^{\mathcal{O}}$ be the subtree of $\overline{T}_{k}$ made of all the vertices in $\mathcal{C}_k$ and the edges connecting them. 
For any $\ell,b \in \mathbb{Z}_{+}$, a rooted tree $T'$ is said~\cite{mossel2001} to be a \textbf{$(b,\ell)$-diluted tree} if: for all $i \in \mathbb{Z}_{+}$, each of the vertices of $T'$ at graph distance $i \ell$ from the root has at least $b$ descendants at graph distance $i\ell + \ell$ from the root.
From the following lemma adapted from \cite{mossel2001}, $\overline{T}_k^{\mathcal{O}}$ is $(2,3)$-diluted with arbitrarily high probability provided edge lengths are short enough. 
\begin{figure}
    \centering
    \includegraphics[scale=0.75]{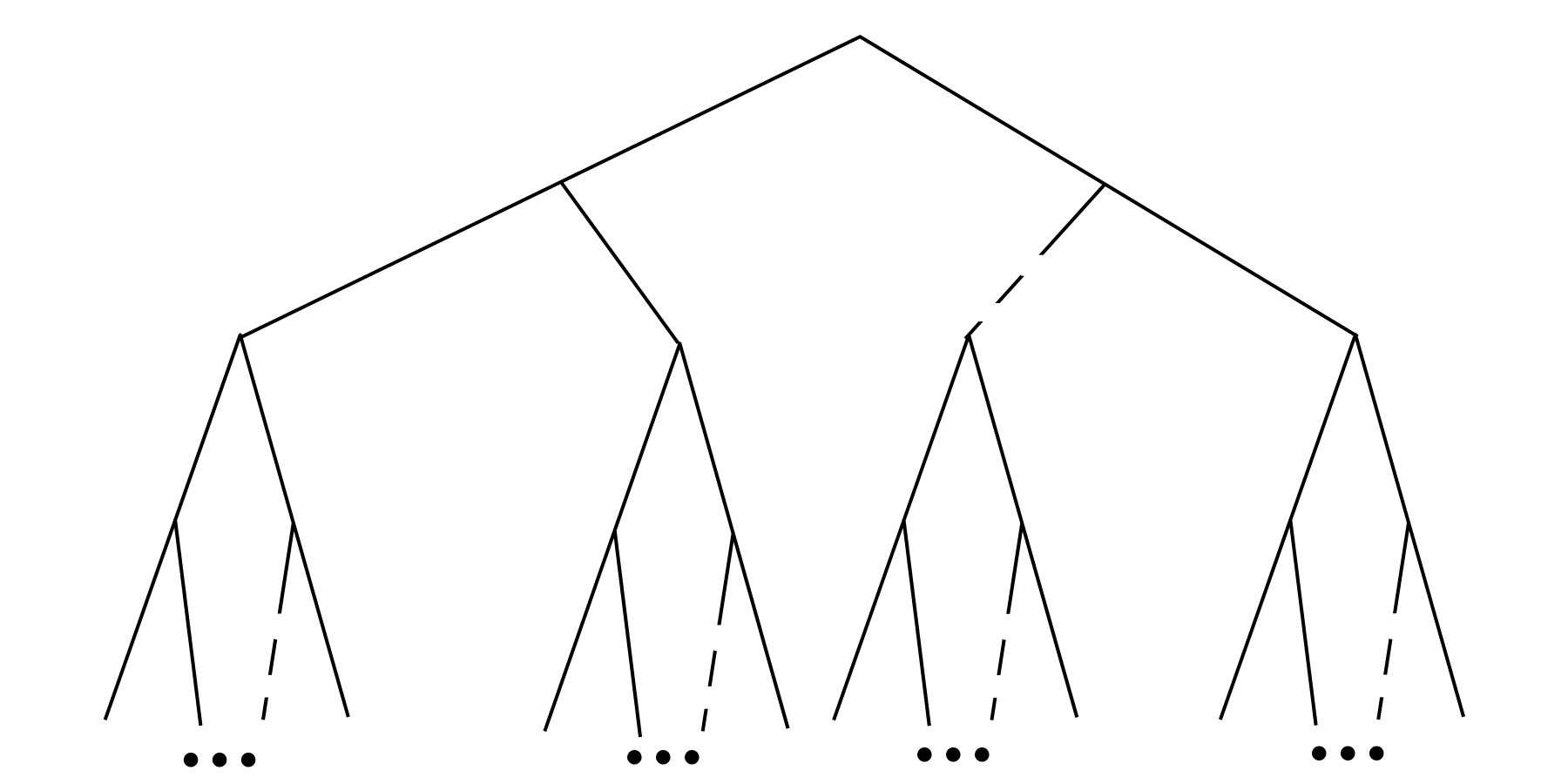}
    \caption{An open $2$-diluted $3$-regular subtree of the infinite binary rooted tree.  Solid lines not descending from any dashed line indicate no mutation.  Dashed lines indicate a mutation may have occurred.  For every vertex in the even generations, at least three of its grandchildren share the same trait.}
    \label{fig:DilutedTree}
\end{figure}
Figure \ref{fig:DilutedTree} depicts the infinite $2$-diluted $3$-regular tree.
 Later in the proof, we will need to condition on the length of $\sigma_{z_k}$ being less than a 
 threshold $\bar{L}$.  Let $\PP^{\bar{L}}$ be the probability measure of the joint process where $\sigma_{z_k}$ is drawn from the stationary distribution of the TKF91 process conditioned on $|\sigma_{z_k}| \leq \bar{L}$.
We first record a simple observation. 
\begin{lemma}[Staying probability]
\label{lemma:staying}
Fix $\bar{L} < +\infty$. 
For any sequence $\sigma$ such that $|\sigma| \leq \bar{L}$ and any $t > 0$, we have
$$
P_t(\sigma, \sigma)
\geq 1 - t (\bar{L}+1)(\mu + \lambda + \eta).
$$
\end{lemma}
\begin{proof}
Indeed
$P_t(\sigma, \sigma)$
is lower bounded by the probability 
that no mutation occurs up to time $t$
which, by~\eqref{eq:lambdastar}, is at least
$$
P_t(\sigma, \sigma)
\geq \exp\left(
-(\bar{L}+1)[\mu + \lambda + \eta] t
\right)
\geq 1 - (\bar{L}+1)[\mu + \lambda + \eta] t,
$$
as claimed.
\end{proof}
\begin{lemma}[Existence of an open diluted tree]
\label{lem:Diluted}
For any $\bar{L} \in \mathbb{Z}_+$ and
$\delta_a > 0$, there is $\tmax > 0$ small enough that, if $t_e \leq \tmax$ for all $e$,
$$
\PP^{\bar{L}}[\overline{T}_k^{\mathcal{O}}\ \textnormal{is $(2,3)$-diluted}] 
\geq 1 - \delta_a.
$$
\end{lemma}
\begin{proof}
By Lemma~\ref{lemma:staying}, for any sequence $\sigma_{z_k}$ such that $|\sigma_{z_k}|\leq \bar{L}$, we have that
$$
\zeta^{(k)}_{e} 
= P_{t_e}(\sigma_{z_k}, \sigma_{z_k})
\geq 1
-(\bar{L}+1)[\mu + \lambda + \eta] \tmax,
$$
where we recall that $t_e \leq \tmax$ by assumption (and that of course includes the added edges of length $0$).
Hence $\zeta^{(k)}_{e}$ can be made arbitrarily close to $1$ (uniformly in $e$) by taking $\tmax$ small enough
(as a function of $\bar{L}$). 
The result then follows directly from \cite[Lemma 8]{mossel2001} (which can be extended in a straightforward manner to the case where percolation probabilities vary across edges but are uniformly bounded).
\end{proof} 

\subsection{Analyzing the Fitch estimator}

Next, we analyze the Fitch estimator 
in the event that $\overline{T}_k$ contains an open $(2,3)$-diluted subtree.  

For any $D \in \mathbb{Z}_+$, let $\overline{T}_{k,D}$ 
be the truncation of $\overline{T}_{k}$ at level $D$, that is, the finite tree obtained by removing all vertices of $\overline{T}_{k}$ at graph distance greater than $D$ from its root. Let $\beta_k$ be the smallest positive integer such that
$\overline{\overline{T}}_k := \overline{T}_{k,2 \beta_k}$ contains all of $T_k$.
Importantly, we make the following observation about the Fitch estimator. 
\begin{lemma}[Fitch estimator on the completion]
\label{lemma:fitch-completion}
The Fitch estimator applied to the leaves of $\overline{\overline{T}}_k$ produces the same ancestral sequence estimate as the Fitch estimator applied to the leaves of $T_k$.
\end{lemma}
\begin{proof}
All leaves $\bar{\bar{\ell}}$ of $\overline{\overline{T}}_k$ descending from a leaf $\ell$ of $T_k$ satisfy $\sigma_{\bar{\bar{\ell}}} = \sigma_\ell$, so by definition of the Fitch estimator $\hat{S}_\ell = \sigma_\ell$. The claim follows. 
\end{proof}

Let $\overline{\overline{T}}_k^{\mathcal{O}}$ be
the truncation of $\overline{T}_k^{\mathcal{O}}$ at level $2 \beta_k$.
\begin{lemma}[Fitch estimator in the presence of an open diluted tree]
\label{lemma:fitch-diluted}
If $\overline{\overline{T}}_k^{\mathcal{O}}$ 
is $(2,3)$-diluted, then the Fitch estimator $\hat{\sigma}_{z_k}$ over the tree $\overline{\overline{T}}_k$ 
equals the true sequence $\sigma_{z_k}$ 
at $z_k$.
\end{lemma}
\begin{proof}
We prove this claim by induction on $\beta_k$.  We start with the $\beta_k = 1$ case.  Then $\overline{\overline{T}}_k$ 
consists of $z_k$, two children $z_k^1$ and $z_k^2$, and the grandchildren $z_k^{1,1},z_k^{1,2},z_k^{2,1},z_k^{2,2}$.  If all four grandchildren belong to $\overline{\overline{T}}_k^{\mathcal{O}}$, 
then we are done.  
The other case, without loss of generality, is $z_k^{2,1} \notin \overline{\overline{T}}_k^{\mathcal{O}}$.  Then $\sigma_{z_k^{2,1}} \ne \sigma_{z_k^{2,2}} = \sigma_{z_k}$, so the Fitch method gives $\hat{S}_{z_k^{2}} = \hat{S}_{z_k^{2,1}} \cup \hat{S}_{z_k^{2,2}} = \{\sigma_{z_k^{2,1}},\sigma_{z_k}\}$.  Since $\sigma_{z_k^{1,1}} = \sigma_{z_k^{1,2}} = \sigma_{z_k}$, we have $\hat{S}_{z_k^{1}} = \{\sigma_{z_k}\}$.  Continuing on, we have $\hat{S}_{z_k} = \hat{S}_{z_k^{1}} \cap \hat{S}_{z_k^2} = \{\sigma_{z_k}\}$.  Since $\hat{S}_{z_k}$ contains only the state $\sigma_{z_k}$, the Fitch method is guaranteed to return $\sigma_{z_k}$.

Now, we assume the $r$-th case holds for $r \geq 1$ and we show that the $(r+1)$-st case holds as well.  As before, consider the four grandchildren of $z_k$ and the same cases.  If all four grandchildren belong to $\overline{\overline{T}}_k^{\mathcal{O}}$, 
then they are each the root of a subtree of $2r$ levels with root state equal to $\sigma_{z_k}$.  The induction assumption implies that the Fitch method returns $\sigma_{z_k}$ as estimates for $\sigma_{z_k^{i,j}}, i,j \in \{1,2\}$.  The Fitch method then returns $\hat{S}_{z_k} = \{\sigma_{z_k}\}$, as required.  
For the other case when $z_k^{2,1} \notin \overline{\overline{T}}_k^{\mathcal{O}}$, 
we know only that $\hat{S}_{z_k^{2,1}}$ 
is an arbitrary set of sequences.  
If $\sigma_{z_k} \in \hat{S}_{z_k^{2,1}}$, 
then $\hat{S}_{z_k^{2}} = \hat{S}_{z_k^{2,1}} \cap \hat{S}_{z_k^{2,2}} = \{\sigma_{z_k}\}$, 
and we are done. 
Else, we have $\hat{S}_{z_k^2} = \hat{S}_{z_k^{2,1}} \cup \hat{S}_{z_k^{2,2}}$, where $\sigma_{z_k} \in \hat{S}_{z_k^{2,2}}$ so that $\hat{S}_{z_k} = \hat{S}_{z_k^1} \cap \hat{S}_{z_k^2} = \{\sigma_{z_k}\}$, as required.  
This completes the proof for the $(r+1)$-st case, 
and hence of the lemma.  
\end{proof}

Combining Lemmas~\ref{lem:Diluted}, \ref{lemma:fitch-completion}, 
and 
\ref{lemma:fitch-diluted},
we get the following.
\begin{prop}[Correctness of ancestral estimation off the backbone]
\label{prop:Fitch2}
For any $\bar{L} \in \mathbb{Z}_+$ and
$\delta_a > 0$, there is $\tmax > 0$ small enough that,
under $\PP^{\bar{L}}$, the Fitch estimator
on $T_k$ returns the correct ancestral state $\hat\sigma_{z_k} = \sigma_{z_k}$ with probability
at least $1-\delta_a$.
\end{prop}

\section{One-mutation condition}
\label{section:one-mutation}

In this section, we establish the one-mutation
condition required by Proposition~\ref{prop:correct-align}
and use it to finish the proof of the main result.

\subsection{A bound on the transition probabilities}

We will need a bound on the probability that 
at most one mutation occurs on an edge
along the backbone. Because the state space
of the sequence process is infinite, the rates are unbounded and we state
the next bound explicitly in terms of the length of the sequence
at the start of the edge. Later on, we will use the fact that the length is stationary to control it.
\begin{lemma}[At most one mutation]
\label{lemma:atmostone}
Fix $\bar{L} < +\infty$. 
For any sequence $\sigma$ such that $|\sigma| \leq \bar{L}$ and any $t > 0$, we have
$$
P_t(\sigma, Y_\sigma)
\geq 1 - \left\{t (\bar{L}+2)[\mu+\lambda+\eta]\right\}^2,
$$
where $Y_{\sigma} = \{\sigma\} \cup \mathcal{S}_1(\sigma)$ are the sequences at most one mutation away from $\sigma$.
\end{lemma}
\begin{proof}
For a TKF91 process on an edge started at $\sigma$, let $X_s \in \mathcal{S}$ be the sequence observed at time $s \in [0,t]$ and $T_i$ be the time of the $i$th jump from one state to another state. Then
$$
P_t(\sigma, Y_\sigma)
\geq \PP_\sigma[T_2 > t],
$$
as the event on the right-hand side guarantees a single 
jump, which in turn guarantees that $X_t \in Y_\sigma$. Here $\PP_\sigma$ indicates that the
 edge process is started at $\sigma$. Letting 
$$
f_{T_1|\sigma}(s) 
= \lambda^{\ast}(\sigma) \exp\left(-s \lambda^{\ast}(\sigma)\right),
\qquad
F_{T_1|\sigma}(s) 
= 1 - \exp\left(-s \lambda^{\ast}(\sigma)\right),
$$
be the probability density function and cumulative distribution function of the time of the first jump
started at $\sigma$, we get by the strong Markov property
\begin{align*}
\PP_\sigma[T_2 \leq t]
&= \int_{0}^t
f_{T_1|\sigma}(s)
\sum_{\tau \in \mathcal{S}_1(\sigma)}
\frac{Q(\sigma,\tau)}{\lambda^{\ast}(\sigma)}
F_{T_1|\tau}(t-s) \,\mathrm{d} s\\
&\leq \int_{0}^t
\lambda^{\ast}(\sigma) \exp\left(-s \lambda^{\ast}(\sigma)\right)
\max_{\tau \in \mathcal{S}_1(\sigma)}
\left\{1- \exp\left(-(t-s) \lambda^{\ast}(\tau)\right)\right\}
\,\mathrm{d} s.
\end{align*}
Under the assumption that $|\sigma| \leq \bar{L}$,
it holds that $|\tau| \leq \bar{L}+1$
for any $\tau \in \mathcal{S}_1(\sigma)$,
and hence $\max\{\lambda^{\ast}(\sigma), \lambda^{\ast}(\tau)\} \leq (\bar{L}+2)[\mu+\lambda+\eta]$. Continuing on,
the last line in the previous display is
\begin{align*}
&\leq \left\{1- \exp\left(-t (\bar{L}+2)[\mu+\lambda+\eta]\right)\right\}
\int_{0}^t
\lambda^{\ast}(\sigma) \exp\left(-s \lambda^{\ast}(\sigma)\right)
\,\mathrm{d} s\\
&= \left\{1- \exp\left(-t (\bar{L}+2)[\mu+\lambda+\eta]\right)\right\}
\left\{1 - \exp\left(-t \lambda^{\ast}(\sigma)\right)\right\}\\
&\leq \left\{t (\bar{L}+2)[\mu+\lambda+\eta]\right\}^2,
\end{align*}
establishing the claim.
\end{proof}

\subsection{Union bound over the backbone}


We define a number of events whose joint occurrence
guarantees the success of our alignment procedure:
\begin{itemize}
    \item \textit{(One-mutation condition)} For $o = -,+$ and $k = 1,\ldots,B^o-1$, let $F^o_k$ be the event that the sequences at $\tilde{x}^o_k$ and the backbone child vertex of $\tilde{x}^o_k$ (i.e., $\tilde{x}^o_{k+1}$) satisfy constraints (A), (B), or (C) from Section~\ref{section:stepwise}.  

    \item \textit{(Ancestral reconstruction)} For $o = -,+$ and $k = 1,\ldots,B^o-1$, let $G^o_k$ be the event that there is no mutation between the sequences at $\tilde{x}^o_k$ and its off-backbone child vertex $\tilde{z}^o_k$ \textit{and} that $\sigma_{\tilde{z}^o_k}$ is correctly reconstructed by applying the Fitch method on the subtree rooted at $\tilde{z}^o_k$.
    
    \item \textit{(Root segment)} For $o = -,+$, let $H^o$ be the event that the sequences at the root and at $\tilde{x}^o_1$ are identical.
    
\end{itemize}

The following proposition provides a requirement on the maximum branch length $\tmax$ for all the above events to occur simultaneously. Define the bad event
$$
\mathcal{B} 
= (H^-)^c \cup (H^+)^c \cup \left\{\bigcup_{o=-,+} \bigcup_{k=1}^{B^o-1} (F^o_k)^c \cup (G^o_k)^c\right\}.
$$
Recall that the pre-processing procedure has a  parameter $\delta_1$. 
\begin{prop}[Union bound over the backbone]
\label{prop:Intersection}
Fix a tree height $h > 0$. For any $0 < \delta_1 < h$, there is a $\tmax$ small enough
that
$$
\PP\left[ \mathcal{B} \right] \leq C h \delta_1 \log^2(\delta_1^{-1}),
$$ 
where $C$ is a constant depending only on $\lambda,\mu,\eta$.
\end{prop}

\begin{proof}
We take a union bound over the events making up $\mathcal{B}$.

\paragraph{Controlling the lengths} For each event, we first apply the law of total probability to control for the length of the starting sequence as follows. Suppose that sequence $\tau$ is stationary,
which we denote by $\tau \sim \Pi$. Using the stationary distribution for the length (i.e.,~\eqref{eq:LengthStationary}), we have 
$$
\PP_{\tau \sim \Pi} \left[|\tau| > \bar{L}\right] 
= \sum_{M=\bar{L}+1}^{\infty} 
\left(1 - \frac{\lambda}{\mu}\right) \left(\frac{\lambda}{\mu}\right)^{M} 
= \left(\frac{\lambda}{\mu}\right)^{\bar{L} + 1}.
$$  
The expression on the right is made less than $\delta_1^2$ by choosing 
\begin{equation}\label{eq:barLdef}
\bar{L} 
= \bigg\lceil \frac{\log(\delta_1^{2})}{\log(\lambda/\mu)}\bigg\rceil
\leq C' \log(\delta_1^{-1}),
\end{equation}
for a constant $C' > 0$ depending only on $\mu, \lambda$,
where recall that $\mu > \lambda$.
Then for any event $\mathcal{E}$ which depends on
$\tau$, we can write
\begin{align}
\PP[\mathcal{E}]
&= \PP[\mathcal{E}\,|\,|\tau| \leq \bar{L}] \,\PP[|\tau| \leq \bar{L}]
+ \PP[\mathcal{E}\,|\,|\tau| > \bar{L}] 
\,\PP[|\tau| > \bar{L}]\nonumber\\
&\leq \PP[\mathcal{E}\,|\,|\tau| \leq \bar{L}]
+ \PP[|\tau| > \bar{L}]\nonumber\\
&\leq \PP[\mathcal{E}\,|\,|\tau| \leq \bar{L}]
+ \delta_1^2,\label{eq:controllingLength}
\end{align}
for the choice of $\bar{L}$ above.

\paragraph{Events $H^o$}
For $o = -,+$, we use Lemma~\ref{lemma:staying}
to bound the probability of $(H^o)^c$. By construction,
$\tilde{x}^o_1$ is a child of the root, so the edge length between the root and $\tilde{x}^o_1$ is at most $\tmax$. Hence, using Lemma~\ref{lemma:staying} and~\eqref{eq:controllingLength} with $\tau := \sigma_\rho$ and $\mathcal{E} := (H^o)^c$,
we get
\begin{equation}
\label{eq:boundH}
\PP[(H^o)^c]
\leq \tmax (\bar{L}+1)(\mu + \lambda + \eta)
+ \delta_1^2.
\end{equation}

\paragraph{Events $G^o_k$}
For $o = -,+$ and $k = 1,\ldots,B^o-1$, we use
Lemma~\ref{lemma:staying} together with
Proposition~\ref{prop:Fitch2} to bound
the probability of
$(G^o_k)^c$. Here we take
$\tau := \sigma_{\tilde{x}^o_k}$ and $\mathcal{E} := (G^o)^c$. By assumption, the edge length between $\tilde{x}^o_k$ and its off-backbone child $\tilde{z}^o_k$ is at most $\tmax$.
Further, for any fixed failure probability $\delta_a > 0$ and length threshold $\bar{L}$, the maximum branch length $\tmax$ can be taken small enough for Proposition~\ref{prop:Fitch2} to hold.  
By~\eqref{eq:controllingLength},
we get
\begin{align*}
\PP[(G^o_k)^c]
&\leq \PP[(G^o_k)^c\,|\,|\sigma_{\tilde{x}^o_k}| \leq \bar{L}] + \delta_1^2\\
&\leq \PP\left[\{\sigma_{\tilde{x}^o_k} \neq \sigma_{\tilde{z}^o_k}\}
\bigcup
\left\{\{\sigma_{\tilde{x}^o_k} = \sigma_{\tilde{z}^o_k}\} \cap \{\hat\sigma_{\tilde{z}^o_k} \neq \sigma_{\tilde{z}^o_k}\}\right\}
\,\middle|\,|\sigma_{\tilde{x}^o_k}| \leq \bar{L}\right] + \delta_1^2\\
&\leq \PP\left[\{\sigma_{\tilde{x}^o_k} \neq \sigma_{\tilde{z}^o_k}\}
\,\middle|\,|\sigma_{\tilde{x}^o_k}| \leq \bar{L}\right]\\
&\qquad + \PP\left[
\left\{\{\sigma_{\tilde{x}^o_k} = \sigma_{\tilde{z}^o_k}\} \cap \{\hat\sigma_{\tilde{z}^o_k} \neq \sigma_{\tilde{z}^o_k}\}\right\}
\,\middle|\,|\sigma_{\tilde{x}^o_k}| \leq \bar{L}\right] + \delta_1^2
\end{align*}
We use the Markov property to bound
the second term as follows:
\begin{align*}
&\PP\left[
\left\{\{\sigma_{\tilde{x}^o_k} = \sigma_{\tilde{z}^o_k}\} \cap \{\hat\sigma_{\tilde{z}^o_k} \neq \sigma_{\tilde{z}^o_k}\}\right\}
\,\middle|\,|\sigma_{\tilde{x}^o_k}| \leq \bar{L}\right]\\
&= \PP\left[
\sigma_{\tilde{x}^o_k} = \sigma_{\tilde{z}^o_k}
\,\middle|\,|\sigma_{\tilde{x}^o_k}| \leq \bar{L}\right]
\,\PP\left[
\hat\sigma_{\tilde{z}^o_k} \neq \sigma_{\tilde{z}^o_k}
\,\middle|\,\sigma_{\tilde{x}^o_k} = \sigma_{\tilde{z}^o_k}, |\sigma_{\tilde{x}^o_k}| \leq \bar{L}\right]\\
&= \PP\left[
\sigma_{\tilde{x}^o_k} = \sigma_{\tilde{z}^o_k}
\,\middle|\,|\sigma_{\tilde{x}^o_k}| \leq \bar{L}\right]
\,\PP\left[
\hat\sigma_{\tilde{z}^o_k} \neq \sigma_{\tilde{z}^o_k}
\,\middle|\,|\sigma_{\tilde{z}^o_k}| \leq \bar{L}\right]\\
&\leq \PP\left[
\hat\sigma_{\tilde{z}^o_k} \neq \sigma_{\tilde{z}^o_k}
\,\middle|\,|\sigma_{\tilde{z}^o_k}| \leq \bar{L}\right].
\end{align*}
Plugging this back above and using Lemma~\ref{lemma:staying} and
Proposition~\ref{prop:Fitch2}
gives
\begin{equation}
    \label{eq:boundG}
\PP[(G^o_k)^c]
    \leq \tmax (\bar{L}+1)(\mu + \lambda + \eta)
+ \delta_a + \delta_1^2.
\end{equation}

\paragraph{Events $F^o_k$}
For $o = -,+$ and $k = 1,\ldots,B^o-1$, we use
Lemma~\ref{lemma:atmostone} to bound
the probability of
$(F^o_k)^c$. Here we take
$\tau := \sigma_{\tilde{x}^o_k}$ and $\mathcal{E} := (F_k^o)^c$. By construction (i.e., by the backbone sparsification pre-processing step), the edge length between $\tilde{x}^o_k$ and its backbone child $\tilde{x}^o_{k+1}$ is at most $2 \delta_1$.
By~\eqref{eq:controllingLength}
and Lemma~\ref{lemma:atmostone}, we get
\begin{align}
\PP[(F^o_k)^c]
&\leq \PP[(F^o_k)^c\,|\,|\sigma_{\tilde{x}^o_k}| \leq \bar{L}] + \delta_1^2\nonumber\\
&\leq \left\{2 \delta_1 (\bar{L}+2)[\mu+\lambda+\eta]\right\}^2 + \delta_1^2.\label{eq:boundF}
\end{align}

\paragraph{Union bound}
Taking a union bound over all events above
gives
\begin{align*}
\PP[\mathcal{B}]
&\leq 2 \left[\tmax (\bar{L}+1)(\mu + \lambda + \eta)
+ \delta_1^2\right]\\
&\qquad + \sum_{o=-,+}\sum_{k=1}^{B^o-1} \left[
\tmax (\bar{L}+1)(\mu + \lambda + \eta)
+ \delta_a + \delta_1^2
\right]\\
&\qquad + \sum_{o=-,+}\sum_{k=1}^{B^o-1} \left[
\left\{2 \delta_1 (\bar{L}+2)[\mu+\lambda+\eta]\right\}^2 + \delta_1^2
\right]
\end{align*}
by \eqref{eq:boundH},~\eqref{eq:boundG} and
\eqref{eq:boundF}. 
We make all terms in square brackets of order $\delta_1^2 \log \delta_1^{-1}$ by choosing $\delta_a := \delta_1^2$ and then choosing $0 < \tmax \leq \delta_1^2$ small enough 
for Proposition~\ref{prop:Fitch2} to hold. 
Then
we get, using~\eqref{eq:barLdef},
\begin{align*}
\PP[\mathcal{B}]
&\leq 2 \left[\delta_1^2 ( C' \log(\delta_1^{-1})+1)(\mu + \lambda + \eta)
+ \delta_1^2\right]\\
&\qquad + 2 (B^o-1) \left[
\delta_1^2 ( C' \log(\delta_1^{-1})+1)(\mu + \lambda + \eta)
+  2\delta_1^2
\right]\\
&\qquad + 2 (B^o-1) \left[
\left\{2 \delta_1 ( C' \log(\delta_1^{-1})+2)[\mu+\lambda+\eta]\right\}^2 + \delta_1^2
\right].
\end{align*}
Because the tree has height $h$ and each backbone edge has length at least $\delta_1$ (after pre-processing), with the exception of the first and last one on each side of the root, we must have $(B^o-2) \delta_1 \leq h$, or after rearranging $B^o \leq h/\delta_1 + 2$. Employing this bound and simplifying gives finally
\begin{equation*}
\PP[\mathcal{B}]
\leq C h \delta_1 \log^2(\delta_1^{-1}),
\end{equation*}
for a constant $C$ depending only on $\mu, \lambda, \eta$,
as claimed.
\end{proof}

\subsection{Proof of the theorem}

We are now ready to finish the proof of the main result.

\begin{proof}[Proof of Theorem~\ref{thm:main}]
For a fixed failure probability $\varepsilon$, we first choose $\delta_1$ small enough (as a function of $h, \mu, \lambda, \eta$) such that $C h \delta_1 \log^2(\delta_1^{-1}) \leq \varepsilon$. 
We then choose $\tmax$ small enough (again as a function of $h, \mu, \lambda, \eta$) 
that Proposition~\ref{prop:Intersection} implies $\PP[\mathcal{B}] \leq \varepsilon$.  
Proposition~\ref{prop:correct-align} then completes the proof of the theorem.
\end{proof}

\section*{Acknowledgments}

SR is grateful to Alexandre Bouchard-C\^ot\'e (UBC) for insightful discussions in the early stages of this project.

SR was supported by NSF grants DMS-1614242, DMS-1902892, DMS-1916378 
and DMS-2023239 (TRIPODS Phase II), as well as a Simons Fellowship and a Vilas Associates Award. BL was supported by NSF grants DMS-1614242, CCF-1740707 (TRIPODS), DMS-1902892 and a Vilas Associates Award (to SR).  BL was also supported by NSF grant DMS-1646108 (University of Michigan-Ann Arbor, Department of Statistics, RTG).

\newpage
\bibliographystyle{alpha}
\bibliography{ASRbib.bib}

\newpage
\appendix

\section{Further lemmas}

We formally justify Eq.~\eqref{eq:lambdastar}.
\begin{lemma}[Total rate]
The rate matrix $Q$ satisfies
\begin{equation*}
\sum_{\tau \in \mathcal{S}_1(\sigma)}Q(\sigma,\tau) = \lambda (|\sigma| + 1) + \mu |\sigma| + \eta \pi_1 |\sigma|_0 + \eta \pi_0 |\sigma|_1.
\end{equation*}
\end{lemma}
\begin{proof}
The transition rate matrix $Q$ entries are as follows.  For the transition from $\sigma$ to $\tau$, where $\tau$ differs from $\sigma$ only by a single substitution, we have 
\begin{align*}
    Q(\sigma,\tau) = \eta \pi_i,
\end{align*} 
when label $i \in \{0,1\}$ is the product of the single substitution.
The sum over all $\tau$ in $\mathcal{S}_{s}(\sigma)$, the sequences that differ by $\sigma$ by a single substitution, is 
\begin{align}
    \label{eqn:SubstitutionSum}
    \sum_{\tau \in \mathcal{S}_{s}(\sigma)}Q(\sigma,\tau) = \eta \pi_1 |\sigma|_0 + \eta \pi_0 |\sigma|_1.
\end{align}  
Next, for the transitions from $\sigma$ to $\tau$ in $\mathcal{S}_{d}(\sigma)$, the sequences obtained by deleting a single digit of $\sigma$, we have
\begin{align*}
    Q(\sigma,\tau) = s_{\sigma,\tau}^{(D)}\mu,
\end{align*} 
where $s_{\sigma,\tau}^{(D)} \geq 1$ is the length of the repeated segment of letters where the deletion occurs.  Here, we observe that the position of the deletion within the segment is not needed.  This is because the letters are identical within the segment, so the remaining letters are identical when one is deleted. For example, if a zero is deleted from a string of $10$ zeros, then $s_{\sigma,\tau}^{(D)} = 10$.  Since only one deletion out of $|\sigma|$ can occur, the sum is
\begin{align}
\label{eqn:DeletionSum}
    \sum_{\tau \in \mathcal{S}_{d}(\sigma)}Q(\sigma,\tau) = |\sigma|\mu.
\end{align} 
Finally, for the transitions from $\sigma$ to $\tau$ in $\mathcal{S}_{i}(\sigma)$, the sequences obtained by inserting a digit into $\sigma$, we have 
\begin{align*}
    Q(\sigma,\tau) = (s_{\sigma,\tau}^{(I)}+1) \lambda \pi_i, \ i \in \{0,1\}
\end{align*} 
where $i$ is the digit inserted and $s_{\sigma,\tau}^{(I)} \geq 0$ is the size of the repeated segment of digit $i$ where the site is being inserted.  Noting that only one insertion can occur and the immortal link might insert a site, the sum is
\begin{align}
    \label{eqn:InsertionSum}
    \sum_{\tau \in \mathcal{S}_{i}(\sigma)}Q(\sigma,\tau) = (|\sigma| + 1)\lambda.
\end{align} 
The claim follows from combining~\eqref{eqn:SubstitutionSum},
\eqref{eqn:DeletionSum}
and~\eqref{eqn:InsertionSum}.
\end{proof}

\end{document}